\documentclass[10pt]{amsart}
\usepackage{amsfonts}
\usepackage{amssymb}
\usepackage{amsmath,amsthm}
\usepackage{amscd}

\makeatletter
\def\ps@pprintTitle{%
 \let\@oddhead\@empty
 \let\@evenhead\@empty
 \def\@oddfoot{}%
 \let\@evenfoot\@oddfoot}
\makeatother

\newcommand{\beq}{\begin{equation}}
\newcommand{\eeq}{\end{equation}}
\newcommand{\bea}{\begin{eqnarray}}
\newcommand{\eea}{\end{eqnarray}}
\newcommand{\bc}{\begin{cases}}
\newcommand{\ec}{\end{cases}}
\newcommand{\nn}{\nonumber}
\newcommand\noi{\noindent}

\newtheorem{definition}{Definition}
\newtheorem{proposition}{Proposition}
\newtheorem{theorem}{Theorem}

\theoremstyle{definition}

\newtheorem{remark}{\textbf{Remark}}

\begin{document}

\author{Piergiulio Tempesta}
\address{Departamento de F\'{\i}sica Te\'{o}rica II (M\'{e}todos Matem\'{a}ticos de la F\'isica), Facultad de F\'{\i}sicas, Universidad
Complutense de Madrid, 28040 -- Madrid, Spain \\ and
Instituto de Ciencias Matem\'aticas, C/ Nicol\'as Cabrera, No. 13--15, 28049 Madrid, Spain}
\email{p.tempesta@fis.ucm.es, piergiulio.tempesta@icmat.es}
\title[Groups and entropies]{Formal Groups and Z-Entropies}
\date{November 15, 2016}
\maketitle


\begin{abstract}
We shall prove that the celebrated R\'enyi entropy is the first example of a new family of infinitely many multi-parametric entropies. We shall call them the \textit{Z-entropies}. Each of them, under suitable hypotheses,  generalizes the celebrated entropies of Boltzmann and R\'enyi.

A crucial aspect is that every $Z$-entropy is \textit{composable} \cite{Tempesta6}. This  property means that the entropy of a system which is composed of two or more independent systems depends, in all the associated probability space, on the choice of the two systems only. Further properties are also required, to describe the composition process in terms of a group law.

The composability axiom, introduced as a generalization of the fourth Shannon-Khinchin axiom (postulating additivity),  is a highly non-trivial requirement. Indeed, in the trace-form class, the Boltzmann entropy and Tsallis entropy are the only known composable cases. However, in the non-trace form class, the Z-entropies arise as new entropic functions possessing the mathematical properties necessary for information-theoretical applications, in both classical and quantum contexts.

From a mathematical point of view, composability is intimately related to formal group theory of algebraic topology. The underlying group-theoretical structure determines crucially the statistical properties of the corresponding entropies.
\end{abstract}


\tableofcontents

\section{Introduction}

Since the pioneering work by Boltzmann, Clausius and Gibbs, the notion of entropy has been widely investigated for its prominence in thermodynamics and statistical mechanics, in both classical and quantum contexts \cite{BS}--\cite{Landsberg}. The research activity inaugurated by Shannon, Khinchin and R\'enyi recognized the foundational role of entropy in modern information theory \cite{Renyi}--\cite{Khinchin}. In particular, new entropic functions, designed for different purposes, have been introduced (see e.g. \cite{Tbook}).

The purpose of this paper is to show that the theory of generalized entropies can be mathematically interpreted, and widely extended, by means of an approach based on \textit{formal group theory} \cite{Haze}. Since the seminal paper by Bochner \cite{Boch}, formal groups have been intensively investigated in the last decades, because they play a prominent role in several branches of mathematics, especially algebraic topology \cite{BMN}--\cite{Faltings}, combinatorics \cite{AB}, the theory of elliptic curves \cite{Serre2}, arithmetic and analytic number theory \cite{MT}--\cite{Tempesta5}.

In order to relate the notion of entropy to group theory, we shall discuss the mathematical requirements that an entropic function has to satisfy. A natural set of conditions is represented by the first three Shannon-Khinchin (SK) axioms (see appendix A for their formulation). These axioms correspond to the requirements of continuity, expansibility (adding an event of zero probability does not affect an entropy) and of the maximization of entropy on the uniform distribution.

At the same time, another crucial property that, in our opinion, any entropy must satisfy is that of \textit{composability} \cite{Tempesta6}. Essentially, it requires that the entropy of a system obtained by composing two independent systems $A$ and $B$ can  always be expressed in terms of the entropies of $A$ and of $B$ only, \textit{for any possible choice of the probability distributions of $A$ and $B$}.

This property is at the heart of the notion of entropy. Indeed, it is related  to the requirement, pointed out in \cite{Gell-Mann}, that entropy is usefully defined on macroscopic states of a given system, without the need for any knowledge of the underlying microscopical dynamics. 

As clarified in \cite{Tempesta6}, composability is even more demanding: we should be able to compose two independent systems in a commutative way, and \textit{three} independent systems in an associative way. Also, if we compose a system with another one in a zero-entropy configuration, the entropy of the compound system should be equal to the entropy of the first system. In other words, a group-theoretical structure is needed.
All these requirements were encoded in \cite{Tempesta6} in the \textit{composability axiom}. It generalizes the fourth Shannon-Khinchin axiom (additivity of an entropy) and allows a large family of entropic functions to be constructed.

Concerning the mathematical form that generalized entropies can assume, we distinguish two different classes.

The first one is the \textit{trace-form class}. It contains entropies of the form $\sum_{i=1}^{W} F(p_{i})$, where $\{p_i\}_{i\in\mathbb{N}}$ is a probability distribution and $F(p_i)$ is a suitable function. The standard case of the Boltzmann-Gibbs entropy is recovered when $F(p_i)=p_i\ln\left(\frac{1}{p_i}\right)$. In the last three decades, this family of entropic functions has been widely investigated, due to its relevance in the theory of complex systems.

The second one is the class of \textit{non-trace-form entropies}. The most well-known representative is the R\'enyi entropy
$
S_{\alpha}=\frac{\ln \sum_{i=1}^{W} p_{i}^{\alpha}}{1-\alpha}$,  $\alpha>0 .
$
It was introduced in \cite{Renyi1960}  and has inspired much work in information theory. Its quantum version is crucial in the study of entanglement of quantum complex systems.

A huge part of the set of entropies introduced in the last three decades, perhaps because of the influence of the classical Boltzmann entropy, belongs to the trace-form class. Also, these new entropies are supposed to recover Boltzmann's entropy in some appropriate limit. Nevertheless, the trace-form class has a serious drawback: the generic lack of composability, with the two remarkable exceptions of the Boltzmann entropy and the Tsallis entropy (which in this paper we present in a two-parametric form).


However, many generalized trace-form entropies are composable in a weak sense only, i.e. over the \textit{uniform probability distribution}. This feature of \textit{weak composability} is certainly necessary for thermodynamic purposes,  but its lack of generality makes it unsatisfactory.
However, one of the outcomes of our analysis is that by renouncing the trace-form structure  we can define new composable entropies. To summarize, the most relevant result of this work is the following.

\textbf{Main result}.
\textit{There exists a new family of non-trace-form strictly composable entropies}.


Precisely, we shall prove that the celebrated R\'enyi entropy is the simplest representative of a huge family of entropies that we shall call the \textbf{Z-entropies}. Each of them depends on a real parameter $\alpha$ and possesses a group structure.  Indeed, to each entropy there is associated a suitable invertible function $G$, allowing under certain hypotheses a \textit{formal group law} to be defined \cite{Haze}, i.e. a two-variable series determined by the relation $\Psi(x,y)=G(G^{-1}(x)+G^{-1}(y))$  which satisfies the axioms of an Abelian group. The group law is nothing but the functional equation obeyed by the considered $Z$-entropy under the composition of two statistically independent systems.
In this framework, the R\'enyi entropy has associated the function $\Psi(x,y)=x+y$, obtained when $G=Id$, i.e. the additive group law. The corresponding functional equation is $S_{\alpha}(A\cup B)= S_{\alpha}(A)+ S_{\alpha}(B)$, for any couple of independent subsystems $A$, $B$. This construction enables us to formulate the whole theory of generalized entropies into a firm foundational group-theoretical framework.

\noi A very general expression of the non-trace-form family of group entropies is (see Definition \ref{ZG} below)
\beq
Z_{G,\alpha}(p_1,\ldots,p_W):=\frac{\ln_{G} \left(\sum_{i=1}^{W} p_{i}^{\alpha}\right)}{1-\alpha}, \label{Zent}
\eeq
where $\ln_{G}(x)$ is a generalized logarithm and $\alpha>0$ is a real parameter. When $0<\alpha<1$, the $Z_{G,\alpha}$ entropy is concave; when $\alpha>1$, under mild hypotheses it is Schur concave. Throughout this paper, we shall assume that the Boltzmann constant $k_{\mathcal{B}}=1$.

A fundamental property of the $Z$-entropies \eqref{Zent} is that, under suitable hypotheses, they \textit{generalize, at the same time, the entropies of Boltzmann and R\'enyi} in a unified expression. 
In a specific example, the Tsallis entropy (jointly with those of Boltzmann and R\'enyi) is also recovered.

By choosing adequately the generalized group logarithm in eq. \eqref{Zent}, we can also define \textit{strictly composable analogues} of well known entropies, including those of Kaniadakis, Borges-Roditi, etc. None of them is indeed composable.


In the trace-form case, the group-theoretical approach to the concept of entropy  has been already formulated in \cite{Tempesta6}. There, the notion of \textit{universal-group entropy}, as the entropy related to the Lazard universal formal group, has been introduced.
%

From the previous considerations, it emerges that there exists a \textit{dual construction} among the two classes of entropies. Precisely, each group law admits (at least) two realizations in terms of two different entropic forms, belonging to the two classes $Z_{G, \alpha}[p]$ and $S_{U}[p]$, respectively.

The additive group has two representatives: the Boltzmann entropy and the R\'enyi entropy. The multiplicative group $\Psi(x,y)= x+y+ \gamma xy$, where $\gamma\in\mathbb{R}$, has interesting realizations, as the Tsallis entropy in the trace-form case, and the entropy \eqref{AT}
in the non-trace-form family. 

There are several reasons to introduce the class of $Z$-entropies. A first motivation comes from the need to establish a mathematically unifying approach to the theory of generalized entropies. From this point of view, the \textit{group entropies}, namely entropies satisfying the first three SK axioms and possessing a group-theoretical structure ensuring the validity of the composition process over all possible states, seem to possess a special role.
From a mathematical point of view, each group entropy of the $Z$-class satisfies two crucial properties: strict composability and \textit{extensivity}. Indeed, one can show that under mild conditions, for a given phase space growth rate $W=W(N)$, there exists a specific entropy of the $Z$-class such that, over the uniform distribution, $\lim_{N\to\infty} S(N)/N= \lambda>0$. A priori this allows, at least formally, one to develop (in a micro-canonical picture) a formalism similar to the classical one for ergodic systems. However, the purpose of this work is not to develop thermodynamics. The full clarification of the possible thermodynamic meaning of the class of group entropies (as already pointed out in \cite{Tempesta6}, \cite{Tempesta4}, \cite{Tempestaprepr}) is an open problem.

Another reason comes from the study of quantum entanglement. Indeed, quantum versions of $Z$-entropies (briefly discussed in this work) could be very useful in the detection of entanglement for quantum complex systems, as spin chains. An example is proposed  where a $Z$-entropy can be used, instead of the quantum R\'enyi entropy, as a measure of the ground state entanglement for a generalized Lipkin-Meshkov-Glick model that was recently introduced \cite{CFGRT}. In particular, this new quantum entropy is extensive (i.e. proportional to $L$) in a context when the entropy of R\'enyi is not.

At the same time, the $Z_{G, \alpha}$ entropies are possibly relevant in information theory. Indeed, they widely generalize both the classical $\alpha$-divergences and the R\'enyi divergences. One can also expect that $Z$-entropies could play an important role in information geometry.

To conclude, a few words concerning the organization of the article. In section 2, the group-theoretical apparatus necessary for the subsequent discussion is sketched. In section 3, we formulate the composability axiom.  In section 4, generalized logarithms are derived from group laws. In section 5, we discuss properties of the trace-form class and more specifically of Tsallis entropy, which is the most general composable trace-form entropy. The new family of $Z$-entropies is defined  in section 6 and some of its properties are proved in section 7. In section 8, some representatives of the $Z$-family are discussed in detail. In particular, the relevant example of the $Z_{a,b}$-entropy is presented in section 9. A quantum version of the $Z$-class is proposed in section 10. Some open problems are proposed in the final section 11.

\section{Groups and entropies: a general approach}
In the subsequent Sections, we shall present a comprehensive theory of generalized entropies based on formal group laws. We will start by recalling some basic facts and definitions of the theory of formal groups (see \cite{Haze} for a thorough exposition, and \cite{Serre} for a shorter introduction).

\subsection{The Lazard  formal group}

Let $R$ be a commutative associative ring  with identity, and $R\left\{ x_{1},\text{ }%
x_{2},..\right\} $ be the ring of formal power series in the variables $x_{1}$, $x_{2}$,
... with coefficients in $R$.

\begin{definition} \label{formalgrouplaw}
A commutative one--dimensional formal group law
over $R$ is a formal power series $\Psi \left( x,y\right) \in R\left\{
x,y\right\} $ such that \cite{Boch}
\begin{equation*}
1)\qquad \Psi \left( x,0\right) =\Psi \left( 0,x\right) =x,
\end{equation*}%
\begin{equation*}
2)\qquad \Psi \left( \Psi \left( x,y\right) ,z\right) =\Psi \left( x,\Psi
\left( y,z\right) \right) \text{.}
\end{equation*}
When $\Psi \left( x,y\right) =\Psi \left( y,x\right) $, the formal group law is
said to be commutative.
\end{definition}
The existence of an inverse formal series $\omega
\left( x\right) $ $\in R\left\{ x\right\} $ such that $\Psi \left( x,\omega
\left( x\right) \right) =0$ is a consequence of Definition \ref{formalgrouplaw}. Let  $B=\mathbb{Z}[b_{1},b_{2},...]$ be the  ring of integral polynomials in infinitely many variables. We shall consider the series
\beq
F\left( s\right) = \sum_{i=0}^{\infty} b_i \frac{s^{i+1}}{i+1}
\label{I.1},
\eeq
with $b_0=1$. Let $G\left( t\right)$ be its compositional inverse:
\beq
G\left( t\right) =\sum_{k=0}^{\infty} a_k \frac{t^{k+1}}{k+1} \label{I.2},
\eeq
i.e. $F\left( G\left( t\right) \right) =t$. From this property, we deduce $a_{0}=1, a_{1}=-b_1, a_2= \frac{3}{2} b_1^2 -b_2,\ldots$.
The Lazard formal group law \cite{Haze} is defined by the formal power series
\[
\Psi_{\mathcal{L}} \left( s_{1},s_{2}\right) =G\left( G^{-1}\left(
s_{1}\right) +G^{-1}\left( s_{2}\right) \right).
\]
The coefficients of the power series $G\left( G^{-1}\left( s_{1}\right) +G^{-1}\left(
s_{2}\right) \right)$ lie in the ring $B \otimes \mathbb{Q}$ and generate over $\mathbb{Z}$ a subring $L \subset B \otimes \mathbb{Q}$, called the Lazard ring.

For any commutative one-dimensional formal group law over any ring $R$, there exists a unique homomorphism $L\to R$ under which the Lazard group law is mapped into the given group law (the \textit{universal property} of the Lazard group).

Let $R$ be a ring with no torsion \footnote{The torsion-free hypothesis was implicit also in \cite{Tempesta6} and \cite{Tempestaprepr}.}. Then, for any commutative one-dimensional formal group law $\Psi(x,y)$ over $R$, there exists a series $\psi(x)\in R[[x]] \otimes \mathbb{Q}$ such that
\[
\psi(x)= x+ O(x^2), \quad \text{and} \quad \Psi(x,y)= \psi^{-1}\left(\psi(x)+\psi(y)\right)\in R[[x,y]]\otimes \mathbb{Q}.
\]

In the subsequent considerations, the universal formal group will play the role of a very general composition law admissible for the construction of the entropies of the $Z$-family.

\section{The Composability Axiom}
In this section we shall first define the notion of composability, in the strict and weak sense, by following (and slightly generalizing) \cite{Tempesta6}.
\begin{definition} \label{composab}
An entropy $S$ is strictly  (or strongly) \textit{composable} if there exists a continuous function of two real variables $\Phi(x,y)$ such that

\text{(C1)}
 \beq
S(A \cup B)=\Phi(S(A),S(B);\{\eta\}) \label{C1}
\eeq
where $A$ and $B$ are two statistically independent systems, defined for any probability distribution $\{p_{i}\}_{i=1}^{W}$, $\{\eta\}$ is a possible set of real continuous parameters,  with the further properties

(C2) Symmetry:
\beq
\Phi(x,y)=\Phi(y,x). \label{C2}
\eeq

(C3) Associativity:
\beq
\Phi(x,\Phi(y,z))=\Phi(\Phi(x,y),z). \label{C4}
\eeq

(C4) Null-composability:
\beq
\Phi(x,0)=x. \label{C3}
\eeq
\end{definition}
This axiom is necessary to ensure that a given entropy may be suitable for thermodynamic purposes. Indeed, it should be obviously symmetric with respect to the exchange of the labels $A$ and $B$. At the same time, if we compose a given system with another system in a state with zero entropy, the total entropy should coincide with that of the given system. Finally, the composability of more than two independent systems in an \textit{associative} way is crucial to define a zeroth law.

\begin{definition} A \textbf{\textit{group entropy}} is a function $S(p_1,\ldots,p_W)\rightarrow R^{+}\cup\{0\}$ that satisfies the first three SK axioms and is strictly composable.
\end{definition}

\begin{remark} \label{remark1}
Observe that in the class of formal power series, the function $\Phi(x,y)$,  satisfying the conditions \eqref{C1}-\eqref{C4}, according to Definition \ref{formalgrouplaw} is a formal \textit{group law} over the reals, with a suitable formal series inverse $\varphi(x)$ such that $\Phi(x,\varphi(x))=0$ \cite{Haze}. The general form of this group law is
$
\Phi(x,y)=x+y+\sum_{kl} c_{kl} x^k y^l.
$
\end{remark}

\noi
Other properties of $\Phi(x,y)$ can be found in \cite{Tempesta6}, where the notion of \textit{weak composability} was also proposed.
The weak formulation essentially requires the composability of the generalized logarithm associated with a given entropy.  The weak property is satisfied by almost all of the generalized entropies proposed in the literature. However, it is not obvious: Indeed, there are entropic forms that satisfy the first three SK axioms, but are not weakly composable.

In \cite{Tempesta4}, this notion (called there ``composability" \textit{tout court}) has been formulated for a class of entropic functions coming from difference operators.


\section{Generalized logarithms and exponentials from group laws}
There is a simple construction allowing  a generalized logarithm from a given group law to be defined. Needless to say, many other definitions of generalized logarithm can be proposed. The following one is of a group-theoretical nature.
\begin{definition} \label{defglog}
A  generalized group logarithm is a continuous,  strictly increasing and strictly concave function $\ln_{G}: (0,\infty)\to \mathbb{R}$, with $\ln_{G}(1)=0$ (possibly depending on a set of real parameters); the functional equation
\beq
\ln_{G}(xy)= \chi(\ln_{G}(x),\ln_{G}(y)) \label{glog}
\eeq
will be called the group law associated with $\ln_{G}(\cdot)$.
\end{definition}
\par
Observe that the eq. \eqref{glog} has a simple realization
\beq
\chi(x,y):=G(G^{-1}(x)+G^{-1}(y)), \label{GPsi}
\eeq
being $G(x):=\ln_{G}(exp(x))$ a strictly increasing continuous function, vanishing at zero.

One of our key results is the following simple proposition, which is a restatement of the previous observation.
%

\begin{proposition}
Let $G$ be a continuous strictly increasing function vanishing at zero. The function $F_{G}(x)$ defined by
\beq
F_{G}(x):= G\left(\ln x^{\gamma}\right), \qquad \gamma\neq 0 \label{Glog}
\eeq
satisfies  eq. \eqref{glog}, where $\chi(x,y)$ is the group law \eqref{GPsi}.
\end{proposition}
\begin{proof}
It suffices to observe that
\bea \label{proofTh1}
F_{G}(xy)&=& G\left(\ln x^{\gamma}+\ln y^{\gamma}\right)=G\left(G^{-1}(F_{G}(x))+G^{-1}(F_{G}(y))\right)\\
\nn &=& \chi(F_{G}(x),F_{G}(y)).
\eea 
\end{proof}

This result can also be formulated in a field of characteristic zero for $G$ in the class of formal power series.  


\noi By way of an example, when $\chi(x,y)=x+y$, we have directly that $G(t)=t$ and $\ln_{G}(x)=\ln x$. If
$\chi(x,y)=x+y+(1-q)xy$, an associated function $G(t)$ is provided by $G(t)=\frac{e^{(1-q)t}-1}{1-q}$ and the group logarithm converts into the Tsallis logarithm
\beq
\ln_q(x):=\frac{x^{1-q}-1}{1-q}. \label{Tslog}
\eeq
\begin{remark}
Let $G$ be a strictly increasing real analytic function of the form \eqref{I.2}. For  $\ln_{G}(x)=G(\ln x)$, the requirement of concavity is guaranteed, for instance, by the simple (but quite restrictive) condition
\beq
 a_{k} >    (k+1)  a_{k+1} \qquad \forall k\in\mathbb{N}, \qquad \text{with } \{a_{k}\}_{k\in\mathbb{N}} > 0 \label{conc},
\eeq
which is also sufficient to ensure that the series $G(t)$ is convergent absolutely and uniformly over the compacts with a radius $r=\infty$. Many other choices are allowed.

\end{remark}

Thus, given a group law, under mild hypotheses we may determine a generalized group logarithm, for instance by means of relation \eqref{Glog} and the condition \eqref{conc} (see \cite{Tempesta4} for a construction of group logarithms from difference operators via the associated group exponential $G$).

\begin{definition} \label{defgexp}
The inverse of a generalized group logarithm will be called the associated generalized group exponential. \end{definition}
This function can be represented in the form
\beq
exp_{G}(x)= e^{G^{-1}(x)}. \label{Gexp}
\eeq

When $G(t)=t$, we have returned to the standard exponential; when $G(t)=\frac{e^{(1-q)t}-1}{1-q}$, we recover the $q$-exponential $e_{q}(x)=\left[1+(1-q)t\right]^{\frac{1}{1-q}}$, and so on.
\begin{remark}
From a computational point of view, observe that the formal compositional inverse $G^{-1}(s)$, such that $G(G^{-1}(s))=s$ and $G^{-1}(G(t))=t$, can be obtained by means of the Lagrange inversion theorem. We get the formal power series
\beq
G^{-1}(s)=s-\frac{a_1}{2}s^2+ \ldots.
\eeq
Also, given $\Psi(x,y)$, by imposing the relation $\Psi(x,y)= G\left(G^{-1}(x)+G^{-1}(y)\right)$, with $\Psi(x,y)=x+y+ \text{higher order terms}$, we get a system of equations that usually allows us to reconstruct the sequence $\{a_{k}\}_{k\in\mathbb{N}}$. Each element $a_{k}$  a priori depends on the set of parameters appearing in $\Psi(x,y)$.
\end{remark}

\section{On the trace-form class of entropies}

\subsection{Basic Properties}


Let us denote by  $N$ the number of particles of a complex system.
\begin{definition} \label{extensivity}
An entropic function $S\left(p_1,\ldots,p_W\right)$ is said to be extensive over the uniform distribution if there exists a function (``growth rate" or ``occupation law" of phase space) $W=W(N)$ such that
\[
\lim_{N\to\infty} S(N)/N = \lambda \ .
\]
\end{definition}
Here $\lambda$ may depend on thermodynamic variables but not on $N$.

\begin{definition} \label{traceform}
An entropy $S$ belongs to the trace-form class if it can be written in the form $S[p]=\sum_{i=1}^{W}F(p_i)$,  where $F$ is a real continuous function, strictly concave, with $F(0)=F(1)=0$  and $F\left([0,1]\right) \subset \mathbb{R}^{+}\cup \{0\}$.
\end{definition}


\subsection{A general logarithm for a composable entropy}

The equation for the multiplicative formal group allows us to give a two-parametric preentation of the logarithm \eqref{Tslog}. Indeed, we observe that the functional equation
\[
f(x+y)=f(x)+f(y)+(1-q)f(x)f(y)
\]
$q\neq 1$, admits the solution
\[
f(x)=\frac{e^{a (1-q) x}-1}{1-q}.
\]
This is coherent with the general theory of functional equations \cite{Aczel}. From this observation, we can also formulate a two-parametric presentation of the Tsallis entropy $S_q$, physically equivalent to it, whose main interest in this context is that it is still \textit{strictly composable}.
To this aim, we consider the function defined by
\beq
S_{a,q}:=\sum_{i=1}^{W} p_i \ln_{q} \left(\frac{1}{p_{i}^{a}}\right)=\frac{1-\sum_{i=1}^{W}p_{i}^{a(q-1)+1}}{q-1}, \label{Sqa}
\eeq
for $a> 0$, $q\in\mathbb{R}/\{1\}$, with the condition $a(q-1)+1>0$.
Needless to say, the previous constraints in the space of parameters ensure that the entropy $S_{a,q}$ satisfies the first three SK axioms \textit{for all real values} of $q$ and specific values of the parameter $a$ (for $a=1$, $q\to 1$ we recover the Boltzmann entropy). Note that the entropies $S_{a,q}$ and $S_{q}$ are related by means of the simple formulae
\[
S_{a,q}=a S_{q'}, \qquad q'=a(q-1)+1 \ .
\]
A direct analysis allows us to establish the following result.
\begin{proposition}
The entropy \eqref{Sqa} is concave in the regions:

i) \[ q\in(-\infty,1), \qquad 0<a< \frac{1}{1-q} \]

ii) \[ q\in (1, \infty), \qquad a>0 . \]
\end{proposition}
It is easy to show that the entropy $S_{a,q}$ is extensive (on equal probabilities) for $W(N)\sim N^{\rho}$. Indeed, the value $q^{*}=1-\frac{1}{a \rho}$ is found such that $\lim_{N\to\infty} S_{a,q}(N)/N= \lambda$,  in a suitable range of $a$, for values of $\rho > 1$.

\noi One also obtains the following interesting property.
\begin{proposition}
Let $A$ and $B$ be two independent systems. The entropy $S_{a,q}$ is strictly composable, with composition law given by
\[
S_{a,q}(A \cup B)= S_{a,q} (A)+S_{a,q} (B)+(1-q) S_{a,q} (A)S_{a,q} (B).
\]
\end{proposition}
The entropy $S_{a,q}$ represents the general solution of the previous functional equation in the trace-form class.

\noi \textbf{Comment}. It is interesting to observe that the Landsberg-Vedral entropy \cite{LV} $S_{LV}= \frac{S_{q}}{\sum_i p_{i}^{q}}$, also called the normalized Tsallis entropy, satisfies the law $S^{LV}(A\cup B)= S^{LV}(A)+ S^{LV}(B) + (q-1) S^{LV}(A)S^{LV}(B)$. This entropy, therefore, belongs to the family of \textit{composable} entropies. However, it is not of trace-form type.

A fundamental property of the trace-form class is the following.
\begin{remark} \label{BT}
\textit{The Tsallis entropy (with the Boltzmann entropy as its reduction) is the only known entropy of the trace-form class which is strictly composable.}
\end{remark}

\section{New strictly composable entropies}

\subsection{The Z-family}
If we restrict ourselves to the trace-form class, composability is realized essentially in the case of the  Tsallis entropy, with the Boltzmann entropy as its fundamental reduction.


However, if we relax the trace-form requirement, new possibilities arise. Indeed, Renyi's entropy $R_{\alpha}$ does not have a trace-form structure, and is strictly composable (additive). Thus, we shall introduce a new family of group entropies, depending on an entropic parameter $\alpha$, that do not belong to the trace-form class and generalize the  R\'enyi entropy.

\begin{remark}\label{independence}
In the following, we shall often assume for simplicity the hypothesis of generalized group logarithms of the form $\ln_{G}(x) =G(\ln x)$, with $G$ a strictly increasing analytic function of the form \eqref{I.2}, where $a_i$ are parameters which can vanish; we always assume that they are independent of the entropic parameter $\alpha$.
\end{remark}
\begin{definition} \label{ZG}
Let $\{p_i\}_{i=1,\cdots,W}$, $W\geq 1$, with $\sum_{i=1}^{W}p_i=1$, be a discrete probability distribution. Let $\ln_{G}$ be a generalized group logarithm, according to Definition \ref{defglog}. The associated  $Z$-entropy of order $\alpha$ is defined to be the function
\beq
Z_{G,\alpha}(p_1,\ldots,p_W):=\frac{\ln_{G} \left(\sum_{i=1}^{W} p_{i}^{\alpha}\right)}{1-\alpha},
\eeq
with $\alpha>0$, $\alpha\neq 1$.
\end{definition}



\noi The following theorem establishes one of the most relevant properties of the $Z$-class of entropies: its strict composability.

\begin{theorem}
The $Z_{G,\alpha}$-entropy is strictly composable: given any two statistically independent systems $A$, $B$, defined on an arbitrary probability distribution $\{p_i\}_{i=1,\cdots,W}$, it satisfies the composition rule
\beq
Z_{G,\alpha}(A\cup B)= \Phi(Z_{G,\alpha}(A),Z_{G,\alpha}(B)) \label{Zrule}
\eeq
where
\[
\Phi(x,y)= \frac{1}{1-\alpha}\chi\left((1-\alpha)x,(1-\alpha)y\right)
\]
with $\chi\left(x,y\right)$ being the group law satisfied by the generalized group logarithm $\ln_{G}$ associated to $Z_{G,\alpha}$.
\end{theorem}
\begin{proof}
Let $\{p_{i}^{A}\}_{i=1}^{W_{A}}$ and $\{p_{j}^{B}\}_{j=1}^{W_{B}}$ be two sets of probabilities associated with two statistically independent systems $A$ and $B$. The joint probability is given by
$
p_{ij}^{A \cup B}=p_{i}^{A}\cdot p_{j}^{B}.
$
We have
\bea
\nn Z_{G, \alpha}(A \cup B)&=& \frac{\ln_{G} \left(\sum_{i,j} (p_{ij}^{A \cup B})^{\alpha}\right)}{1-\alpha}=\frac{\ln_{G} \left[\left(\sum_{i=1}^{W_A} (p_{i}^{A})^{\alpha}\right)\cdot\left( \sum_{j=1}^{W_B} (p_{j}^{B})^{\alpha} \right)\right]}{1-\alpha} = \\
\nn &=& \frac{\Phi\left(\ln_{G} \left(\sum_{i=1}^{W_A} (p_{i}^{A})^{\alpha}\right),\ln_{G} \left(\sum_{j=1}^{W_B} (p_{j}^{B})^{\alpha}\right)\right)}{1-\alpha}= \\ \nn
&=&\frac{\Phi((1-\alpha) Z_{G,\alpha}(A),(1-\alpha)Z_{G,\alpha}(B))}{1-\alpha}\\ \nn &=&\Psi\left(Z_{G,\alpha}(A),Z_{G,\alpha}(B)\right). \\
\eea
Let us prove that the $Z_{G,\alpha}$-class is strictly composable. It suffices to show that $\Phi(x,y)$ is a group law, as a consequence of the fact that $\chi\left(x,y\right)$  is.
The symmetry property is trivial. The null-composability comes from the relation
\[
\Phi(x,0)=\frac{1}{1-\alpha}\chi\left((1-\alpha)x,0\right)= x.
\]
To prove associativity, observe that
\begin{eqnarray*}
\Phi(\Phi(x,y),z)&=&\frac{1}{1-\alpha} \chi \left[(1-\alpha)\frac{1}{1-\alpha} \chi\left((1-\alpha)x,(1-\alpha)y\right),(1-\alpha)z\right]= \\
&=& \frac{1}{1-\alpha} \chi \left[(1-\alpha)x, \chi\left((1-\alpha)y,(1-\alpha)z\right)\right]= \\
&=& \frac{1}{1-\alpha} \chi\left[(1-\alpha)x, (1-\alpha)\Phi\left(y, z\right)\right]=\Phi(x,\Phi(y,z)).
\end{eqnarray*}
Consequently, the $Z_{G, \alpha}$ entropies are strictly composable.
\end{proof}


\subsection{R\'enyi entropy}
The celebrated R\'enyi entropy
$R_{\alpha}=\frac{\ln \sum_{i=1}^{W} p_{i}^{\alpha}}{1-\alpha}$ is the first example in the $Z$-class. It is obtained by identifying the generalized logarithm with the standard one. 
R\'enyi's entropy  is indeed the first representative of an infinite tower of entropies, that we shall construct in the following sections. It corresponds to a solution to the functional equation $Z_{G,\alpha}(A\cup B)=Z_{G,\alpha}(A)+Z_{G,\alpha}(B)$. 

\section{Main properties of the class of Z-entropies}

\subsection{Limiting properties}

The class of Z-entropies possesses a simple relation to the most celebrated entropies, in particular to Boltzmann's and R\'enyi's entropies, as stated by the following results.
\begin{proposition}
Under the hypotheses of Remark 5, in the limit  $\alpha\to 1$, the $Z_{G,\alpha}$-entropy reduces to the Boltzmann-Gibbs entropy.
\end{proposition}
\begin{proof}
\[
\lim_{\alpha\to 1} \frac{\ln_{G} \left(\sum_{i=1}^{W} p_{i}^{\alpha}\right)}{1-\alpha}=\lim_{\alpha\to 1} \frac{G\left[ \ln\left(\sum_{i=1}^{W} p_{i}^{\alpha}\right)\right]}{1-\alpha}=
\]
\[
=\lim_{\alpha\to 1} \frac{G'\left[ \ln\left(\sum_{i=1}^{W} p_{i}^{\alpha}\right)\right]\cdot \frac{1}{\sum_{i=1}^{W} p_{i}^{\alpha}}\cdot \sum_{i=1}^{W} p_{i}^{\alpha}\ln p_i}{-1}= \sum_{i=1}^{W} p_{i}\ln \frac{1}{p_i}.
\]
where we took into account the fact that $G'(0)=1$. 
\end{proof}
\begin{proposition}
Under the hypotheses of Remark 5, the $Z_{G,\alpha}$-entropy generalizes the R\'enyi entropy.
\end{proposition}
\begin{proof}
It suffices to observe that the function $\ln_{G}(x)$ reduces to $\ln x$ when $a_{i}\to 0$, $i=1,2,\ldots$,  in the expression \eqref{I.2} of $G(t)$. 
\end{proof}

A crucial property of the $Z$-class is that it is compatible with the SK requirements, as made clear by the following statement.

\subsection{The SK axioms}
\begin{theorem}
The $Z_{G,\alpha}$-entropy satisfies the first three SK axioms, for $0<\alpha<1$. \label{SKTh}
\end{theorem}
\begin{proof}

We shall check the validity of the axioms in the formulation given in the appendix.

(SK1). By construction, the function \eqref{Zent} is continuous with respect to all of its arguments $p_i$.

(SK2). The entropy \eqref{Zent} is strictly concave for $0<\alpha<1$. Indeed, given two probability distributions $\textbf{p}^{(1)}=\left(p^{(1)}_{1},\ldots,p^{(1)}_{W}\right)$ and $\textbf{p}^{(2)}=\left(p^{(2}_{1},\ldots,p^{(2)}_{W}\right)$, we have
\[
Z_{G,\alpha}(\lambda \textbf{p}_1 + (1-\lambda) \textbf{p}_2)=\frac{\ln_{G} \left(\sum_{i=1}^{W} \left(\lambda p^{(1)}_i + (1-\lambda) p^{(2)}_i\right)^{\alpha}\right)}{1-\alpha} >
\]
\[
> \frac{\ln_{G}\left(\lambda \sum_{i=1}^{W} (p^{(1)}_i)^{\alpha}+(1-\lambda) \sum_{i=1}^{W} (p^{(2)}_i)^{\alpha}\right)}{1-\alpha} >
\]
\[
> \lambda \frac{\ln_{G}\left( \sum_{i=1}^{W} (p^{(1)}_i)^{\alpha}\right)}{1-\alpha} + (1-\lambda) \frac{ \ln_{G}\left( \sum_{i=1}^{W} (p^{(2)}_i)^{\alpha}\right)}{1-\alpha},
\]
where we took into account that $\ln_{G}$ is by construction strictly concave and increasing, and that the function $\sum_{i=1}^{W} p_{i}^{\alpha}$ is concave for $0<\alpha<1$.

(SK3). By adding an event of zero probability, for $\alpha>0$, $\alpha\neq 1$ we have
\beq
Z_{G,\alpha}(p_1,\ldots,p_W,0=p_{W+1})=\frac{\ln_{G} \left(\sum_{i=1}^{W+1} p_{i}^{\alpha}\right)}{1-\alpha}=\frac{\ln_{G} \left(\sum_{i=1}^{W} p_{i}^{\alpha}\right)}{1-\alpha}.
\eeq
\end{proof}
We conclude that the $Z$-entropies are \textit{group entropies}.
\subsection{Extensivity} \label{extens}



 We shall discuss the extensivity properties of the $Z$-family of entropies over the uniform probability distribution (i.e. $p_{i}=1/W$ for all $i=1,\ldots,W$). In other words, we wish to determine the conditions ensuring that a $Z$-entropy is proportional to the number $N$ of particles of a given system,  in the large $N$ limit, when all states are equiprobable. We have that, for $N$ large (and focusing on the leading constribution to the asymptotic behaviour)
\bea
\nn Z_{G,\alpha}[W]\sim \lambda N &\Longleftrightarrow&  \ln_{G}(W^{1-\alpha}) \sim (1-\alpha) \lambda N \\ &\Longleftrightarrow& W(N) \sim \left[\exp_{G} \left((1-\alpha) \lambda N\right)\right]^{1/(1-\alpha)}. \label{occlaw}
\eea
 The expression for $W(N)$ so determined usually can be computed explicitly.

However, the function $W(N)$ must be  interpretable as a phase space growth rate. A \textit{sufficient} condition is that $W(N)$  be defined for all $N\in\mathbb{N}$ as a real function, with $\lim_{N\to\infty} W(N)=\infty$. These requirements usually restrict the space of allowed parameters of the considered entropy.

Consequently, provided the previous condition is satisfied, the entropies of the $Z$-family are extensive on a specific regime, given by a growth rate $W(N)$.


The previous properties altogether indicate the potential relevance of the notion of $Z$-entropy in thermodynamical contexts and, more generally, in the theory of complex systems.

\subsection{The Schur concavity of Z-entropies}

In Theorem \ref{SKTh}, we have assumed $0<\alpha<1$ to guarantee concavity of the $Z$-family. However, if we take $\alpha>1$, under mild hypotheses this family satisfies the  interesting property of Schur concavity.

\noi We recall here some basic facts about majorization \cite{HLP}. A vector $\textbf{a}\in\mathbb{R}^{n}$ majorizes or dominates another vector $\textbf{b}\in\mathbb{R}^{n}$ (and we write $\textbf{a} \succeq \textbf{b})$ if the following properties are satisfied:
\bea
&\textit{i})&\quad \sum_{i=1}^{k}a_i \geq \sum_{i=1}^{k} b_i, \qquad k=1,\ldots,n-1, \\
&\textit{ii})& \quad \sum_{i=1}^{n}a_i=\sum_{i=1}^{n} b_i,
\eea
where $a_i$ and $b_j$ are the elements of $\textbf{a}$ and $\textbf{b}$ sorted in decreasing order.

\begin{definition} \label{Schur-C}
Given two probability distributions $\textbf{p}:=\{p_i\}_{i=1}^{W}$ and $\textbf{r}:=\{r_{j}\}_{j=1}^{W}$, such that $\textbf{p} \preceq \textbf{r}$, we shall say that an entropy $S$ is Schur concave if $S[\bf{p}]\geq$ $S[\bf{r}]$.
\end{definition}

The following result holds.
\begin{theorem}
Assuming the hypotheses of Remark 5, the $Z_{G,\alpha}$-entropy is Schur concave for $\alpha>1$.
\end{theorem}
\begin{proof}
Since, according to Definition \ref{ZG} and our hypotheses, any entropy of the class \eqref{Zent} is permutation invariant in the variables $\left(p_{1},\ldots,p_{W}\right)$ and its first derivatives exist, we shall use the Schur-Ostrowski criterion for Schur concavity \cite{PPT}.
Precisely, we have to prove that, for any probability distribution $\textbf{p}$,
\beq
\left(p_{i}-p_{j}\right)\left(\frac{\partial Z_{G,\alpha}}{\partial p_i}-\frac{\partial Z_{G,\alpha}}{\partial p_j}\right)\leq 0, \quad 1\leq i\neq j\leq W.
\eeq
The previous relation is indeed equivalent to the following inequality:
\beq
\frac{\alpha}{1-\alpha}\left(p_{i}-p_{j}\right) \ln_{G}' \left(\sum_{i=1}^{W} p_{i}^{\alpha}\right)\left({p_{i}^{\alpha-1}-p_{j}^{\alpha-1}}\right) \leq 0,  \quad 1\leq i\neq j\leq W \ ,
\eeq
which holds due to the properties of $\ln_{G}$. 
\end{proof}

\section{A tower of new entropic forms}

The previous construction allows us to establish a correspondence among trace-form and non-trace-form entropies. Precisely, given a well-defined trace-form entropy, \textit{weakly composable}, we can associate with it a new, non-trace-form but \textit{strictly composable}, generalized entropy. In this way, we can generate a tower of  new examples of group entropies that parallel the well-known entropies, with infinitely many more in addition.

As we have clarified, both the entropies of Boltzmann and R\'enyi  correspond to the additive formal group. Let us consider the first example of nonadditive entropic function of the $Z$-class \eqref{Zent}, naturally associated with the multiplicative group law.
\begin{definition} We shall call the function
\beq
Z_{q,\alpha}=\frac{\ln_{q} \left(\sum_{i=1}^{W} p_{i}^{\alpha}\right)}{1-\alpha} \label{AT}
\eeq
the $Z_{q,\alpha}$-entropy. Here $q>0$, $0<\alpha<1$ and $\ln_{q}(x)$ is given by eq. \eqref{Tslog}.
\end{definition}
\noi The composition law for the entropy \eqref{AT} is
\beq
\nn Z_{q,\alpha}(A\cup B)=Z_{q,\alpha}(A)+Z_{q,\alpha}(B)+ (1-\alpha)(1-q) Z_{q,\alpha}(A) Z_{q,\alpha}(B).
\eeq
The entropy \eqref{AT}, after a re-scaling of the parameters $q,\alpha$, coincides with the Sharma-Mittal entropy \cite{SM}.

\noi A new example can be obtained from the group law
\bea
\nn \chi(x,y,k)&=&x\sqrt{1+k^2 y^2} + y \sqrt{1+ k^2 x^2}=x+y+\frac{k^2}{2}(x y^2+x^2 y)-\frac{k^4}{8}(x y^4+x^4 y)+
\\ \nn &+&\text{higher order terms}.
\eea
The associated logarithm is the Kaniadakis logarithm \cite{Kaniad1}, defined by
\beq
\ln_{k}^{K}(x):=\frac{x^{k}-x^{-k}}{2k},
\eeq
with $-1<k<1$. A new element of the $Z$-family is therefore the following entropy.
\begin{definition}
The function
\beq
Z_{k,\alpha}(p_1,\ldots,p_W):=\frac{\left(\sum_{i=1}^{W} p_{i}^{\alpha}\right)^k -\left(\sum_{j=1}^{W} p_{j}^{\alpha}\right)^{-k}}{2k(1-\alpha)} \label{Zk}
\eeq
will be called the $Z_{k,\alpha}$-entropy, with $0<\alpha<1$, $-1<k<1$.
\end{definition}

The strict composability of the entropy \eqref{Zk} can be directly ascertained from the formula
\[
Z_{k,\alpha}(A\cup B)= Z_{k,\alpha}(A)\sqrt{1+k^2 (1-\alpha)^2 Z_{k,\alpha}(B)^2}+ Z_{k,\alpha}(B)\sqrt{1+ k^2(1-\alpha)^2 Z_{k,\alpha}(A)^2}.
\]

This procedure can be easily extended to infinitely many group laws; correspondingly, it yields a sequence of entropic forms. Another particularly interesting case will be discussed below.

\section{The $Z_{a,b}$-entropy: A generalization of the Boltzmann, Tsallis, R\'enyi, Sharma-Mittal entropies}
\subsection{Definition}
In this section, we present a first example  of a composable, three-parametric entropy. Perhaps the most appealing feature of this entropy is that it generalizes
some of the most important entropies known in the literature, which have played a prominent role in information theory, thermodynamics and generally speaking in applied mathematics.

\begin{definition} The  $Z_{a,b}$-entropy is defined to be the function
\beq
Z_{a,b}(p_1,\ldots,p_W):=\frac{\left(\sum_{i=1}^{W} p_{i}^{\alpha}\right)^a -\left(\sum_{j=1}^{W} p_{j}^{\alpha}\right)^b}{(a-b)(1-\alpha)}.
\eeq
with $0<\alpha<1$, and $a>0$, $b\in\mathbb{R}$ or $a\in\mathbb{R}$, $b>0$, with $a\neq b$.
\end{definition}
\subsection{Limiting properties}
The $Z_{a,b}$-entropy reduces to the Boltzmann entropy in the limit  $\alpha\to 1$. We also have the following notable properties.

\begin{proposition}
The $Z_{a,b}$ entropy reduces to the R\'enyi entropy in the double limit $a\to 0$, $b\to 0$.
\end{proposition}

\begin{proposition}
The $Z_{a,b}$ entropy reduces to the Tsallis entropy in the double limit $b\to 0$, $a\to 1$.
\end{proposition}
\begin{proof}
\[
\lim_{a\to 1}\left(\lim_{b\to 0} \frac{\left(\sum_{i=1}^{W} p_{i}^{\alpha}\right)^a -\left(\sum_{j=1}^{W} p_{j}^{\alpha}\right)^b}{(a-b)(1-\alpha)}\right)=\lim_{a\to 1}\frac{\left(\sum_{i=1}^{W} p_{i}^{\alpha}\right)^a -1}{a(1-\alpha)}= \frac{1-\sum_{i=1}^{W} p_{i}^{\alpha}}{\alpha-1},
\]
which is the expression of the Tsallis entropy, as a function of the parameter $\alpha$. 
\end{proof}

\begin{proposition}
The $Z_{a,b}$-entropy reduces to the $Z_{k,\alpha}$-entropy in the limit $a=-b=k$.
\end{proposition}

\begin{proposition}
The $Z_{a,b}$-entropy reduces to the Sharma-Mittal entropy in the limit $b\to 0$.
\end{proposition}
\subsection{Group-theoretical structure: the Abel formal group laws}
The $Z_{a,b}$-entropy is related to the Borges-Roditi logarithm:
\[
\ln_{BR}(x):=\frac{x^{a}-x^{b}}{a-b}.
\]
If we put $x= e^{u}$, we recover the function known in the mathematical literature as the Abel exponential, given by
\beq
\text{exp}_{Abel}=\frac{e^{au}-e^{bu}}{a-b}=\frac{e^{\gamma u}}{\sqrt{\delta}}\sinh \left(\sqrt{\delta} u\right), \label{Abelexp}
\eeq
where
$
\gamma= (a+b)/2$, $\delta=(a-b)^2/4.$
The Abel exponential \eqref{Abelexp} satisfies the Abel functional equation, as was proven in \cite{Abel}.
The composition laws associated  are remarkable: the \textit{Abel formal group laws}, defined by
\beq
\chi_{\mathcal{A}}(x,y)=x+y+\beta_1 xy+ \sum_{j>i} \beta_i\left(xy^{i}-x^{i}y\right). \label{AFG}
\eeq
The coefficients $\beta_n$ in (\ref{AFG}) can be expressed as polynomials in $a$ and $b$ (see corollary 2.2 of \cite{CJ}):
\beq
\beta_1=a+b, \quad \beta_n=\frac{(-1)^{n-1}}{n!(n-1)}\prod_{\overset{i+j=n-1}{i,j\geq 0}} (ia+jb),\quad n>1.
\eeq
The ring over which the universal Abel formal group law is defined has been studied in \cite{BK,CJ}; the cohomology theory associated with this formal group law has been studied in \cite{BK}  and in \cite{Busato}.

\section{Quantum $Z$-entropies}

The class of $Z$-entropies can be quantized by means of a standard procedure. Precisely, let us consider a quantum system, whose states belong to a finite-dimensional Hilbert space, and are represented by semi-positive definite density matrices $\rho$
(see, e.g., \cite{NC}, \cite{Horodecki}). 
\begin{definition} The family of quantum $Z$-entropies is defined by the formal relation
\beq
Z_{G,\alpha}(p_1,\ldots,p_W):=\frac{\ln_{G} \left(\text{tr} \rho^{\alpha}\right)}{1-\alpha}, \label{qZent}
\eeq
where $\alpha>0$, $\alpha\neq 1$ and $\ln_{G}$ is a generalized group logarithm.
\end{definition}
For each specific entropy, the set of parameters appearing in eq. \eqref{qZent} are supposed to satisfy suitable constraints and, if necessary, certain conventions are tacitly adopted (as $0 \ln  0 =0$ in the case of von Neumann entropy); however, we will not treat in detail these issues here.

We shall discuss now the example of the quantum version of the $Z_{a,b}$-entropy. Similar considerations apply to other cases.

\begin{definition} The quantum $Z_{a,b}$-entropy is defined to be
\beq
Z_{a,b}[\rho]=\frac{\left(\text{tr} \rho^{\alpha}\right)^{a} -\left(\text{tr} \rho^{\alpha}\right)^{b}}{(a-b)(1-\alpha)}. \label{qZab}
\eeq
\end{definition}
One of the most interesting aspects of the family of quantum entropies \eqref{qZent} is that they are directly related to the von Neumann entropy
$
S_{VN}= -\text{tr} \rho \ln \rho
$
and to the quantum R\'enyi entropy
$
S_{\alpha}(\rho)=\frac{1}{1-\alpha} \ln \text{tr} \rho^{\alpha}.
$
\noi Concerning the quantum $Z_{a,b}$-entropy, it possesses other interesting limits. For instance, it can be related to the quantum version of the Sharma-Mittal entropy
\beq
Z_{a,0}[\rho]=\frac{\left(\text{tr} \rho^{\alpha}\right)^a -1}{a(1-\alpha)}, \label{qZab}
\eeq
and leads to a two-parametric generalization of Kaniadakis entropy
\beq
Z_{k,-k}[\rho]=\frac{\left(\text{tr} \rho^{\alpha}\right)^{k} -\left(\text{tr} \rho^{\alpha}\right)^{-k}}{2k(1-\alpha)}. \label{qZab}
\eeq
Consider now the example of a quantum system whose Hilbert space is the tensor product $\mathcal{H}_{A} \times \mathcal{H}_{B}$ of two of its subsystems $A$ and $B$. For instance, we can consider a quantum spin chain of $N$ spins, and a block of $L$ sites ($1\leq L \leq N-1$). Therefore, we can identify $\mathcal{H}_A \times \mathcal{H}_B$ with the Hilbert spaces of the first $L$ and $N-L$ spins. Given the ground state $| \psi \rangle$ of the chain, we can introduce the reduced density matrix $\rho_{L}= \text{tr}_{N-L} | \psi \rangle \langle \psi |$, where $\text{tr}_{N-L}$ denotes the trace over the Hilbert space of the remaining $N-L$ sites. Thus, the previous formulae can be immediately written in terms of the reduced density matrix $\rho_{L}$ to define bipartite entanglement entropies for the considered spin chain.

A priori, due to the multi-parametric nature of these entropies, they can be particularly useful in the study of entangled systems. In this perspective, we believe that $Z_{a,b}[\rho]$ deserves special attention.
An analysis of this aspect is outside the scope of the present paper, and will be discussed elsewhere.

\subsection{Applications: the generalized isotropic Lipkin-Meshkov-Glick model} As an example of application of the entropic forms presented in the previous sections, we shall consider the generalized isotropic Lipkin-Meshkov-Glick model of $N$ interacting particles introduced in \cite{CFGRT}. Its Hamiltonian reads
\[
H=\sum_{i<j} h_{ij}\left(1-S_{ij}\right)+\sum_{k=1}^{m}c_k \left(J^{k}-N h_{k}\right)^2,
\]
where $h_{ij}, c_k >0$, $h_{k}\in\mathbb{R}$ and $J^{k}=\sum_{i} J_{i}^{k}$, $1\leq k \leq m$. Here $J_{i}^{k}$ $(k=1,\ldots, m)$ are operators whose action on the $i$-$th$ site is $J_{i}^{k}=T_{i}^{kk}-T_{i}^{m+1,m+1}$, where $T^{pq}$ is the $(m+1)\times (m+1)$ matrix whose only nonzero element is a 1 in the $p$-th row and $q$-th column. The non-degenerate ground state of the model is a Dicke state of $su(m+1)$ type.

The ground state entanglement entropies for this model can be easily computed in the case of our group entropies, starting from the results of \cite{CFGRT}. There, the reduced density matrix for a block of $L$ spins was computed, when the system is in its ground state. In the large $L$ limit, we obtain the formulae
\begin{eqnarray}
Z_{a,0}[\rho]&=&\frac{\alpha^{\frac{-ma}{2}}\left[2 \pi L (1-\gamma) \prod_{j=1}^{m+1} n_{j} ^{1/m}\right]^{\frac{am(1-\alpha)}{2}}-1}{a(1-\alpha)} \\ &\approx &
\frac{L^{\frac{am(1-\alpha)}{2}}}{a (1-\alpha) \alpha^{\frac{ma}{2}}}\left[2 \pi (1-\gamma) \prod_{j=1}^{m+1} n_{j} ^{1/m}\right]^{\frac{am(1-\alpha)}{2}} + O(1), \label{GSEE}
\end{eqnarray}
where $\alpha$ is the entropic parameter, $\gamma=\lim_{N\to\infty}\frac{L}{N}$, $n_j$ are magnon densities (see \cite{CFGRT} for details).
The relevance of formula \eqref{GSEE} for $Z_{a,0}[\rho]$ is due to the fact that it solves an open problem proposed in \cite{CFGRT}. Indeed, the quantum R\'enyi entropy is not extensive, i.e. proportional to $L$. The quantum Tsallis entropy is extensive, but for $m\geq 3$ only. Observe that, instead, the entropy $Z_{a,0}$ can be made \textit{linear} in $L$. Indeed,
we see immediately that the extensivity requirement implies $\alpha = 1- \frac{2}{am}$, which is satisfied for any $m$, for suitable values of $a$.


\section{Open problems and future perspectives}
The present work represents a first exploration of the mathematical properties of a new, large family of non-trace-form \textit{group entropies} coming from formal group theory via the composability axiom. The underlying group-theoretical structure is responsible for essentially all the relevant properties of this family. 




The results obtained above suggest that the $Z$-class  can be a flexible tool,  offering in a future perspective new insight into different contexts of the theory of classical and quantum complex systems, in ecology and social sciences. For instance, we plan to define generalized diversity indices related to the class. Also, $Z$-entropies are particularly suitable for defining new divergences, generalizing those of R\'enyi \cite{EH} and Kullback-Leibler \cite{KL}  and are potentially relevant in the context of information geometry \cite{AN}.

\par
A directly related expression for $Z$-entropies is given by the formula
\beq
S[p]:=G \left(\frac{\ln \sum_{i=1}^{W} p_{i}^{\alpha}}{1-\alpha}\right) \label{newZ} .
\eeq
Although very similar to the main definition of $Z$-entropies proposed in Definition \ref{ZG}, the class \eqref{newZ} deserves attention for its statistical properties, and will be analyzed elsewhere.
\par

We wish to point out that some of the definitions of the theory can be reformulated in different ways, and some conditions can be relaxed. The present approach has the advantage of providing a large class of group entropies in a simple way; however, the  problem of determining the most general form for group entropies is open.
We shall discuss in detail recent progress on these issues elsewhere.

In our opinion, the crucial role played by the group-theoretical approach in the description of compound systems paves the way towards a re-foundation of the theory of generalized entropies in terms of group entropies.

An open problem is to establish whether the Z-entropies  are \textit{Lesche stable}. We conjecture that, in this respect, the entropies of the class \eqref{Zent} have essentially the same behaviour as the R\'enyi entropy. As was proved in \cite{JA}, from many respects, R\'enyi's entropy can also be considered an observable. 
This conjecture will be thoroughly analyzed in a future work.

We also wish to mention that a generalization of the correspondence among entropies and zeta functions \cite{Tempesta4,Tempesta5,Tempestaprepr} to the case of multi-parametric entropies and multiple zeta values and polylogarithms is an interesting open problem. 



\appendix
\section{The Shannon-Khinchin axioms}
As customary in the literature on generalized entropies, here we re-state  the original requirements of Shannon and Khinchin for an entropy to be admissible, in terms of four requirements.

(SK1) (Continuity). The function $S(p_1,\ldots,p_W)$ is continuous with respect to all its arguments.

(SK2) (Maximum principle).  The function $S(p_1,\ldots,p_W)$ takes its maximum value over the uniform distribution $p_i=1/W$, $i=1,\ldots,W$.

(SK3) (Expansibility). Adding an event of zero probability to a probability distribution does not change entropy: $S(p_1,\ldots,p_W,0)=S(p_1,\ldots,p_W)$.

(SK4) (Additivity). Given two subsystems $A$, $B$ of a statistical system,
\[
S(A \cup B)=S(A)+S(B\mid A).
\]


\vspace{3mm}



\vspace{3mm}

\noi \textbf{Fundings}.

\noi This work has been partly supported by the research project FIS2015-63966, MINECO, Spain, Spain, and by the ICMAT Severo Ochoa project SEV-2015-0554 (MINECO).

\vspace{3mm}



\vspace{3mm}

\textbf{Acknowledgments}.

I wish to thank  J. A. Carrasco, A. Gonz\'alez L\'opez, M. A. Rodr\'iguez and G. Sicuro for useful discussions.

\end{document}